%% file: main.tex
\documentclass[11pt]{article}
\input{header.tex}

\begin{document}

\date{}

\title{On the Optimal Space Complexity of Consensus \\ for Anonymous Processes}

\author{
 Rati Gelashvili\\
  \small MIT\\
  \small gelash@mit.edu\\
 }

\maketitle
\begin{abstract}
The optimal space complexity of consensus in shared memory is a decades-old open problem.
For a system of $n$ processes, no algorithm is known that uses a sublinear number of registers.
However, the best known lower bound due to Fich, Herlihy, and Shavit requires $\Omega(\sqrt{n})$ registers.

The special symmetric case of the problem where processes are anonymous 
  (run the same algorithm) has also attracted attention.
Even in this case, the best lower and upper bounds are still $\Omega(\sqrt{n})$ and $O(n)$.
Moreover,  Fich, Herlihy, and Shavit first proved their lower bound for anonymous processes, 
  and then extended it to the general case.
As such, resolving the anonymous case might be a significant step towards understanding and solving the general problem.

In this work, we show that in a system of anonymous processes, any consensus algorithm satisfying 
  nondeterministic solo termination has to use $\Omega(n)$ read-write registers in some execution.
This implies an $\Omega(n)$ lower bound on the space complexity of deterministic obstruction-free and 
  randomized wait-free consensus, matching the upper bound and closing the symmetric case of the open problem.
\end{abstract}

	
\thispagestyle{empty}

\newpage
\setcounter{page}{1}

\input{intro.tex}
\input{bound.tex}

\input{adaptive.tex}

\section{Acknowledgments}
Support is gratefully acknowledged from the National Science Foundation under grants CCF-1217921, 
        CCF-1301926, and IIS-1447786, the Department of Energy under grant ER26116/DE-SC0008923, 
        and the Oracle and Intel corporations.

The author would like to thank Nir Shavit, Michael Coulombe and Dan Alistarh for helpful conversations and feedback,
  and the anonymous reviewers for their excellent comments.

\bibliographystyle{alpha}
\bibliography{biblio}

\appendix
\input{appendix.tex}

\end{document}

%% file: header.tex
\usepackage[papersize={8.5in,11in},margin=1in]{geometry}
\usepackage{amsthm} 
\usepackage{times}
\usepackage{fancyhdr} 
\usepackage{color} 
\usepackage{changepage}
\usepackage[nofillcomment,linesnumbered,noend,noresetcount,noline]{algorithm2e}

\AtBeginDocument{%
 \abovedisplayskip=2pt plus 1pt minus 1pt
 \abovedisplayshortskip=0pt plus 1pt
 \belowdisplayskip=2pt plus 1pt minus 1pt
 \belowdisplayshortskip=0pt plus 1pt
}
\usepackage[noindentafter]{titlesec}
\titlespacing{\paragraph}{%
  0pt}{
  0.1\baselineskip}{
  1em}
\titlespacing\section{0pt}{8pt plus 1pt minus 1pt}{2pt plus 1pt minus 1pt}
\titlespacing\subsection{0pt}{8pt plus 1pt minus 1pt}{2pt plus 1pt minus 1pt}
\titlespacing\subsubsection{0pt}{8pt plus 1pt minus 1pt}{2pt plus 1pt minus 1pt}
\usepackage{amsmath} 
\usepackage{amssymb} 
\usepackage{graphicx} 
\usepackage{complexity} 
\usepackage{enumitem} 
\usepackage{shorttoc} 
\usepackage{array} 
\usepackage{float}
\usepackage{pdfsync}
\usepackage{verbatim}

\usepackage[pagebackref=true,colorlinks]{hyperref}
\hypersetup{linkcolor=red,filecolor=red,citecolor=red,urlcolor=red}

\newtheoremstyle{slplain}
  {.4\baselineskip\@plus.1\baselineskip\@minus.1\baselineskip}
  {.3\baselineskip\@plus.1\baselineskip\@minus.1\baselineskip}
  {\itshape}
  {}
  {\bfseries}
  {.}
  { }
  {}
\theoremstyle{slplain} 

\newtheorem*{definition*}{Definition}
\newtheorem*{theorem*}{Theorem}
\newtheorem{theorem}{Theorem}[section]
\newtheorem{lemma}[theorem]{Lemma}

\newtheorem{claim}[theorem]{Claim}
\newtheorem{corollary}[theorem]{Corollary}
\newtheorem{definition}[theorem]{Definition}

\makeatletter
\newtheorem*{rep@theorem}{\rep@title}
\newcommand{\newreptheorem}[2]{%
\newenvironment{rep#1}[1]{%
 \def\rep@title{#2 \ref{##1}}%
 \begin{rep@theorem}}%
 {\end{rep@theorem}}}
\makeatother

\newreptheorem{theorem}{Theorem}
\newreptheorem{lemma}{Lemma}
\newreptheorem{claim}{Claim}
\newreptheorem{corollary}{Corollary}

\theoremstyle{definition}

\theoremstyle{remark}

\numberwithin{equation}{section}

\newtheoremstyle{etplain}
  {.0\baselineskip\@plus.1\baselineskip\@minus.1\baselineskip}
  {.0\baselineskip\@plus.1\baselineskip\@minus.1\baselineskip}
  {\itshape}
  {}
  {\bfseries}
  {.}
  { }
  {}

\newcommand{\litlow}[1]{\mathord{\mathcode`\-="702D\sf #1\mathcode`\-="2200}}
\newcommand{\lit}[1]{\ensuremath{\litlow{#1}}}

\newcommand{\namedref}[2]{\hyperref[#2]{#1~\ref*{#2}}}

\newcommand{\sectionref}[1]{\namedref{Section}{#1}}
\newcommand{\appendixref}[1]{\namedref{Appendix}{#1}}
\newcommand{\theoremref}[1]{\namedref{Theorem}{#1}}

\newcommand{\figureref}[1]{\namedref{Figure}{#1}}
\newcommand{\figurerefb}[2]{\hyperref[#1]{Figure~\ref*{#1}#2}}

\newcommand{\lemmaref}[1]{\namedref{Lemma}{#1}}
\newcommand{\claimref}[1]{\namedref{Claim}{#1}}

\newcommand{\equationref}[1]{\hyperref[#1]{(\ref*{#1})}}
\renewcommand{\eqref}{\equationref}

\usepackage[textsize=tiny]{todonotes}

\newcommand{\DEBUG}[1]{}

\renewcommand{\setminus}{-}
\renewcommand{\emptyset}{\varnothing}



%% file: intro.tex
\section{Introduction}
The celebrated Fischer, Lynch and Paterson (FLP) ~\cite{FLP85} result proved that fundamental synchronization tasks 
  including consensus and test-and-set are not solvable in a wait-free manner using read-write registers.
However, the work of Ben-Or~\cite{BenOr83} shows that 
  it is possible to circumvent FLP and obtain efficient distributed algorithms,
  if we relax the problem specification to allow probabilistic termination.
It is also possible to solve these tasks deterministically, but obstruction-free instead of wait-free; 
  it is known how to convert any deterministic obstruction-free algorithm 
  into a randomized wait-free algorithm against an oblivious adversary (see~\cite{GHHW13}).
	
The space complexity of an algorithm is the maximum number of registers used in any execution.
A lot of research has been dedicated to improving the upper and lower bounds on the space complexity for canonical tasks.
For test-and-set, an $\Omega(\log{n})$ lower bound was shown in~\cite{SP89} and independently in~\cite{GW12}.
On the other hand, an $O(\sqrt{n})$ deterministic obstruction-free upper bound was given in~\cite{GHHW13}.
The final breakthrough was the recent obstruction-free algorithm designed by Giakkoupis et al.~\cite{GHHW14},
  with $O(\log{n})$ space complexity, essentially closing the problem\footnote{The space complexity of randomized test-and-set against a strong (adaptive) adversary remains open.}.

For consensus, an upper bound with $n$ registers was long known from~\cite{AH90}.
A lower bound of $\Omega(\sqrt{n})$ by Fich et al.~\cite{FHS98} first appeared in 1993. 
The proof is notorious for its technicality and utilizes a neat inductive combination of covering and valency arguments.  
Another version of the proof appeared in a textbook~\cite{AE14}.
However, a linear lower bound or a sublinear space algorithm has remained elusive to date.

The authors of~\cite{FHS98} conjectured a tight lower bound of $\Omega(n)$.
But the linear lower bound has not been proven even in a restricted, symmetric case, where all processes are anonymous.
In such a system processes can be thought of as running the same code: 
  all processes with the same input start in the same initial state and behave identically.
The same linear upper bound holds for anonymous processes, since a deterministic obstruction free consensus algorithm 
  that uses $O(n)$ registers is known~\cite{GR05}.
Interestingly, the proof in~\cite{FHS98} starts by showing the $\Omega(\sqrt{n})$ lower bound for anonymous processes,
  which is then extended to a much more complex argument for the general case.
Therefore, a linear lower bound in the anonymous setting might prove to be a meaningful step 
  in better understanding and solving the general case of the open problem.

\paragraph{Contribution.}
In this paper we prove the $\Omega(n)$ lower bound in the symmetric (anonymous) case
  for consensus algorithms satisfying the standard \emph{nondeterministic solo termination} property.
Any lower bound for algorithms satisfying the nondeterministic solo termination 
  implies a lower bound for deterministic obstruction-free and randomized wait-free algorithms.
As in~\cite{FHS98, AE14}, the bound is for the worst-case space complexity of the algorithm,
  i.e. for the number of registers used in some execution, regardless of its actual probability.

Our argument relies on a specific class of executions which we call \emph{reserving}, 
  and on the ability to define valency, corresponding to possible return values, for these executions. 
This definition of valency and the ability to cover registers with modified contents by reserved processes 
  greatly simplifies the task of performing an inductive argument.
We hope that these techniques will be useful for future work.

We also demonstrate how the lower bound can be extended to a non-anonymous, adaptive, setting
  where processes come from a very large namespace and the bound depends 
  on the size of the subset of processes that actually participate in the execution.
However, this extension requires additional restrictions on register size and termination,
  and is provided mainly to illustrate an approach.
	
\paragraph{Definitions and Notation.}
We use the standard shared-memory model and similar notation to~\cite{FHS98, AE14}.
We consider anonymous processes and atomic read-write registers.
A process is \emph{covering} a register $R$, if the next step of $p$ can be a write to $R$.
A \emph{block write} of a set of processes $P$ to a set of covered registers $V$ 
  is a sequence of write steps by processes in $P$, 
  where each step is a write to a different register and all registers get written to.
	
In a system of anonymous processes, if a process $p$ in state $s$ performs a particular operation, 
  for any configuration with any process $q$ in the same state $s$, $q$ can also perform the exactly same operation.
Finally, if $p$ and $q$ perform the same operation from the same state with the same outcome (i.e. read the same value), 
  then both $p$ and $q$ end up in the same state after the operation.  
In the context of randomized algorithms, 
  anonymous processes always perform the same operation from the same state 
  (including flipping coins with the same random distribution),
  and end up in identical state if they observe the same results.

A \emph{clone} of a process $p$, exactly as in~\cite{FHS98, AE14}, 
  is defined as another process with the same input as $p$, 
  that shadows $p$ by performing the same operations as $p$ in lockstep,
  reading and writing the same values immediately after $p$, and remaining in the same state,
  all the way until some write of $p$.
Because the system consists of anonymous processes, in any execution with sufficiently many processes,
  for any write operation of $p$,
  there always exists an alternative execution with a clone $q$ that shadowed $p$ all the way until the write.
In particular, in the alternative execution,
  process $q$ \emph{covers} the register and is about to write the value that $p$ last wrote there. 
Moreover, the two executions with or without the clone covering the register 
  are completely indistinguishable to all processes other than the clone itself.

An execution is a sequence of steps by processes and 
  a \emph{solo} execution is an execution where all steps are taken by a single process.
An execution interval is a subsequence of consecutive steps from some execution.  
In the binary consensus problem each participating process starts with a binary input $0$ or $1$, 
  and must return a binary output.
The correctness criterium is that all outputs must be the same and equal to the input of some process.
We say that an execution interval \emph{decides} $0$ (or $1$) 
  if some process returns $0$ (or $1$, respectively) during this execution interval.

A wait-free termination requirement means that each participating process must eventually return an output
  within a finite number of own steps, regardless of how the other processes are scheduled.
The FLP result shows that in the asynchronous shared memory model with read-write registers,
  no deterministic algorithm can solve binary consensus in a wait-free way.
However, it is possible to deterministically solve obstruction-free consensus, i.e. 
  when processes are only required to return an output if they run solo from some configuration.
It is also possible to solve consensus in a randomized wait-free way,
  when processes are allowed to flip random coins and decide their next steps accordingly.
A \emph{nondeterministic solo termination} property of an algorithm means that from each reachable configuration,
  for each process, there exists a finite solo execution by the process where it terminates and returns an output.
We prove our lower bounds for binary consensus algorithms that satisfy this \emph{nondeterministic solo termination} property,
  because both deterministic obstruction-free algorithms and randomized wait-free algorithms fall into this category.

%% file: bound.tex
\section{Space Complexity Lower Bound}
In order to demonstrate our approach, we start by presenting a different proof of 
  the $\Omega(\sqrt{n})$ space lower bound in the anonymous setting.
It uses induction on the number of registers written during an execution,
  as opposed to induction on the tuple of sizes of pending block writes in~\cite{FHS98}.
The proof also has an additional benefit that the use of covering and valency arguments is decoupled.
As usual, we use covering to enforce writing to a new register,
  while a valency argument reminiscent of~\cite{FLP85} ensures that both decision values remain reachable by solo executions.    

Next, building upon this new argument, we prove an $\Omega(n)$ space lower bound for consensus 
  with nondeterministic solo termination in a system of anonymous processes.
There are some significant differences, for instance, 
  the execution is constructed in such a way that after a register is written to, it always remains covered.
Moreover, valency is redefined to account for this specific class of executions.
The rest is induction.
\subsection{A Square-Root Lower Bound}
\label{sec:sqrt}
In this section, we define \emph{valency} as follows.
If there is a solo execution of some process returning $0$ from a configuration, 
  then we call this configuration \emph{0-valent} 
  (and \emph{1-valent} if there is a solo execution of a process that returns $1$).
Solo termination implies that every configuration is $0$-valent or $1$-valent.
Note that unlike the standard definition of valency,
  our definition allows the same configuration to be simultaneously $0$-valent and $1$-valent.
We call such configurations that are both $0$-valent and $1$-valent \emph{bivalent}, and \emph{univalent} otherwise.
Notice that a configuration is bivalent if two solo executions of the same process return different values.
If a configuration is $0$-valent, but not $1$-valent 
  (i.e. no solo execution from this configuration decides $1$), 
  then we call it $0$-univalent, meaning that the configuration is univalent with valency $0$.
Analogously, a configuration is $1$-univalent if it is $1$-valent but not $0$-valent.
 
Observe that if we have at least two processes, then in every bivalent configuration
  we can always find two distinct processes $p$ and $q$, 
  such that there is a solo execution of $p$ returning $0$ and a solo execution of $q$ returning $1$.
This is because either the configuration is bivalent because of solo executions of distinct processes, 
  in which case we are done, 
  or two solo executions of some process return different values, 
  in which case it suffices to consider any terminating solo execution of another process.

For the system of anonymous processes, and a consensus algorithm that uses atomic read-write registers and 
  satisfies the nondeterministic solo termination property, we prove the following statement by induction:
\begin{lemma}
\label{lem:sqrt}
For $r \geq 0$, there exists a system of $\frac{(r-1)r}{2}+2$ anonymous processes, 
  such that for any consensus algorithm, a configuration $C_r$ is reachable by an execution $E_r$ with the following properties:
\begin{itemize}[noitemsep,nolistsep]
\item There is a set $R$ of $r$ registers, each of which has been written to during $E_r$, and 
\item the configuration $C_r$ is bivalent.
\end{itemize}
\end{lemma}
\begin{proof}
The proof is by induction, with the base case $r = 0$.
Our system consists of two processes $p$ and $q$, 
  $p$ starts with input $0$, $q$ starts with input $1$, and $C_0$ is the initial state.
Clearly, no registers have been written to in $C_0$ and bivalency follows by nondeterministic solo termination.

Now, let us assume the induction hypothesis for some $r$ and prove it for $r+1$.
By the induction hypothesis, we can reach a configuration $C_{r}$ using $\frac{(r-1)r}{2}+2$ processes. 
The goal is to use another $r$ processes and extend $C_{r}$ to $C_{r+1}$, 
  completing the proof since $r+\frac{(r-1)r}{2}+2=\frac{r(r+1)}{2} + 2$.

As discussed above, because we have at least $2$ processes and $C_{r}$ is bivalent, 
  there exists a process $p$ and its solo execution $\alpha$ from $C_{r}$ after which $p$ returns $0$ and 
  a process $q \neq p$ and its solo execution $\beta$ from $C_{r}$ after which $q$ returns $1$.\footnote{Alternatively one can say execution $E_{r} \alpha$ ends with $p$ returning $0$ and $E_{r} \beta$ ends with $q$ returning $1$.}  
Recall that $R$ is the set of $r$ registers that were written to in execution $E_{r}$.
For each register in $R$, let a new process clone the process that last wrote to it 
  all the way to covering the register poised to write the same value as present in the register in configuration $C_{r}$.

Let us now apply the covering argument utilizing the clones.
Consider execution $E_{r} \alpha \gamma \beta$, where $\gamma$ is a block write to $R$ by the new clones.
We know that process $p$ returns $0$ after $E_{r} \alpha$.
During its solo execution $\alpha$, process $p$ has to write to a register outside of $R$.
Otherwise, the configuration after $E_r \alpha \gamma$ is indistinguishable from $C_{r}$ to process $q$
  as the values in all registers are the same, and $q$ is still in the same state as in $C_{r}$.
Hence, $q$ will return $1$ after $E_{r} \alpha \gamma \beta$ as it would after $E_r \beta$, 
  contradicting the correctness of the consensus algorithm.
Analogously, process $q$ has to write outside of $R$ during $\beta$.
Let $\alpha = \alpha' w_p \alpha''$, where $w_p$ is the first write of $p$ outside the set of registers $R$,
  and let $\beta = \beta' w_q \beta''$, with $w_q$ being the first write outside of $R$.
Let $\ell$ be the length of $\gamma \beta' w_q$ and 
  $B_i$ be a prefix of $\gamma \beta' w_q$ of length $i$, for all possible $0 \leq i \leq \ell$.

Next, we use a valency argument to reach $C_{r+1}$.
We show that either the configuration reached after $\E_{r} \alpha' \gamma \beta' w_q$, 
  or one of the configurations reached after $\E_{r} \alpha' B_i w_p$ for some $i$,  
  satisfies the properties necessary to be $C_{r+1}$.
Clearly, we have used the right number of processes to reach any of these configurations and 
  $r+1$ registers have been written to while doing so, including $R$ and the register written by $w_p$ or $w_q$.
Thus, we only need to show that one of these configurations is \emph{bivalent}.

Assume the contrary.
The configuration for $i=0$ must be $0$-univalent, since $p$ returns $0$ only throughout $\alpha''$, 
  and we assumed that the configuration is not bivalent.
Similarly, the configuration reached after $\E_{r} \alpha' \gamma \beta' w_q = E_r \alpha' B_{\ell}$ is $1$-univalent.
It is univalent by our assumption and $1$-valent as $q$ running solo returns $1$ through $\beta''$
  ($\alpha'$ does not involve a write outside of $R$ and $q$ cannot distinguish from $E_r \beta' w_q \beta''$).
Because the configuration reached after $E_r \alpha' B_{\ell}$ is $1$-univalent, 
  any terminating solo execution of process $p$ from that configuration must also return $1$.
In particular, every terminating solo execution that starts by $p$ performing its next step $w_p$ returns $1$.
So the configuration reached after $E_r \alpha' B_{\ell} w_p$ must be $1$-univalent:
  solo executions of $p$ return $1$ (some solo execution terminates due to nondeterministic solo execution),
  and it is univalent by our assumption (it is the same as configuration for $i=\ell$).
Therefore, the configuration reached after $\E_{r} \alpha' B_i w_p$ is $0$-univalent for $i=0$ and $1$-univalent for $i=\ell$.
Hence, we can find a switching point for some $i$ and $i+1$,
  where the configuration $X$ reached by $\E_{r} \alpha' B_i w_p$ is $0$-univalent, 
  while the configuration $Y$ reached by $\E_{r} \alpha' B_{i+1} w_p$ is $1$-univalent.
Let $o$ be the extra operation in $B_{i+1}$.

Operation $o$ is not by $p$ and may not be a read or a write to the same register as $w_p$ writes to since
  $p$ would not distinguish between $X$ and $Y$ and would return the same output from both configurations
  through the same solo execution, contradicting the existence of the different univalencies.
Otherwise, operations $w_p$ and $o$ commute.
Let $\sigma$ be a terminating solo execution from $Y$ by the process that performed operation $o$,
  where it returns $1$ due to the univalency of $Y$.
Also consider this process performing its next operation $o$ from $X$.
Since $w_p$ and $o$ commute, and $o$ is not a read, 
  the process cannot distinguish between the resulting configuration and $Y$ and returns $1$ through $\sigma$ as from $Y$.
However, $o \sigma$ is a solo execution from $X$ that returns $1$, contradicting the $0$-univalency of $X$.
The contradiction proves the induction step, completing our induction.
\end{proof}
Notice that for $n$ processes,~\lemmaref{lem:sqrt} directly implies the existence of an execution 
  where $\Omega(\sqrt{n})$ registers are written to, proving the desired lower bound.
\input{linear.tex}

%% file: linear.tex
\subsection{Linear Lower Bound}
\label{sec:linbound}
Consider systems with $n$ anonymous processes and an arbitrary correct consensus algorithm 
  satisfying the nondeterministic solo termination property. 
We will assume that no execution of the algorithm uses more than $n/20$ registers (otherwise, we are trivially done),
  and prove that such an algorithm has to use $\Omega(n)$ registers, which completes the proof.
For notational convenience, let us define $m$ to be $n/20$.

The argument in~\lemmaref{lem:sqrt} relies on a new set of clones in each iteration 
  to overwrite the changes to the contents of the registers made during the inductive step. 
This is the primary reason why we only get an $\Omega(\sqrt{n})$ lower bound.
As the authors of~\cite{FHS98} also mention,
  to get a stronger lower bound we would instead have to reuse existing processes.
In order to do so, these existing processes need to cover the registers in our inductive configurations
  (we must also ensure proper valency conditions on what they are about to write, but let us focus on the covering).
Now, even if we reach such a configuration, 
  during a solo execution interval of some process in the subsequent induction step, all the registers may get written to,
  and we would have to use all the covering existing processes to overwrite the changes.
Therefore, in the next configuration, there is no way to guarantee that the existing processes 
  would still cover various registers.

This is the primary reason why we have to replace solo executions in the proof with a different class of executions
  that we call \emph{reserving}.
Intuitively, reserving executions ensure that for the registers that are written to, 
  some processes are reserved to cover them.
This way, we can have reserved processes cover the registers in subsequent inductive configurations.
Notice that the definition of valency used in the proof of~\lemmaref{lem:sqrt} was based on solo executions.
Thus, we also redefine valency based on reserving executions. 
\subsubsection{Reserving executions}
The following is a formal definition of a reserving execution interval.
\begin{definition}
Let $C$ be some configuration reachable by the algorithm, 
  and let $P$ be a set of at least $m+1$ processes.
We call an execution interval $\gamma$ that starts from configuration $C$ \emph{reserving} from $C$ by $P$ if:
\begin{itemize}[noitemsep, nolistsep]
\item Every step in $\gamma$ is by a process in $P$.  
\item At any time during the execution of $\gamma$:  
  if we let $R_w$ be the set of registers written to so far during $\gamma$, 
then, for each register in $R_w$, 
  there is a \emph{reserved} process $p \in P$ covering that register, one per register. 
\item If a process $p \in P$ returns during $\gamma$ then it does so in the last step of $\gamma$. 
\end{itemize}
\end{definition}
Notice that by definition any prefix of a reserving execution interval is also a reserving execution interval. 
Let $\lit{Res}(C, P)$ be the set of all reserving execution intervals from $C$ by processes in $P$
  that end with a process $p \in P$ returning.
We first show that given sufficiently many processes, such an execution interval exists.
This is essential for defining the valency later.
Recall that we assumed a strict upper bound of $m$ on the number of registers that can ever be written.
\begin{claim}
For any reachable configuration $C$ and a set of at least $m+1$ processes $P$, none of which have returned yet, 
  we have that $\lit{Res}(C, P) \neq \emptyset$.
\label{clm:reserved}
\end{claim}
\begin{proof}
For a given $C$ and $P$, we will prove the claim by constructing a particular reserving execution interval $\gamma$ 
  that ends when some process $p \in P$ returns.
We start with an empty $\gamma$ and continuously extend it.
In the first stage, one by one, for each process $p \in P$: 
\begin{itemize}[noitemsep, nolistsep]
\item Due to the nondeterministic solo termination, there exists a solo execution of $p$ where $p$ returns.
  \begin{itemize}[noitemsep, nolistsep]
  \item If $p$ ever writes to any register during this solo execution, 
    extend $\gamma$ by the prefix of the execution before this write, and move to the next process in $P$.
  \item Otherwise, complete $\gamma$ by extending it with the whole solo execution of $p$.
  \end{itemize}
\end{itemize}
We have finitely many processes and the first stage described above consists of extending the execution interval 
  at most $|P|$ times.
Each time, because of the nondeterministic solo termination for some process $p \in P$, 
  we extend $\gamma$ by a prefix of a finite solo execution of $p$.
Moreover, all operations are reads by processes in $P$, and therefore the prefix of $\gamma$ constructed so far is reserving.

If some process returns in the first stage, the construction of $\gamma$ is complete.
Otherwise, since the first stage is finite, we move on to the second stage described below.
In the configuration after the first stage each of the at least $m+1$ processes in $P$ is covering a register 
  (by their next write operation after the first stage).
From that configuration, the execution interval $\gamma$ is extended by repeatedly doing the following:
\begin{itemize}[noitemsep, nolistsep]
\item[1.] Let $R$ be the set of covered registers by processes of $P$.
  Since $|R| \leq m < |P|$, we can find two processes $p, q \in P$ covering the same register in $R$.
\item[2.] Due to the nondeterministic solo termination, there exists a solo execution of $p$ where $p$ returns.
  \begin{itemize}[noitemsep, nolistsep]
  \item If $p$ ever writes to a register outside of $R$ during this solo execution, 
    extend $\gamma$ by the prefix of the execution before this write, and continue from the first step.
    Notice that at the beginning of the next iteration, process $p$ still covers a register as required.
  \item Otherwise, complete $\gamma$ by extending it with the whole solo execution of $p$.
  \end{itemize}
\end{itemize}
In the second stage, each iteration terminates, since for any process $p \in P$, 
  we can extend by at most the terminating solo execution of $p$, which exists and is finite.
After each iteration, if the construction is not complete, the size of $R$ increases by one.
But there are at most $m$ registers in the system and $|R| \leq m$. 
Thus, after at most $m$ finite extensions, we will complete the construction of $\gamma$ when some process returns.

The execution is reserving because at all times, the registers that were written-to are in $R$.
Moreover, for each register in $R$, there is always a process covering it
  starting from the time it was first covered by some process $p$ in the second step of some iteration 
  all the way until the end of $\gamma$.
\end{proof}
\noindent The next claim follows immediately from the definition of reserving executions.
\begin{claim}
\label{clm:prefix}
Consider a reachable configuration $C$, a set of at least $m+1$ processes $P'$ none of which have returned yet, 
  and another configuration $C'$ reached after some process $p \not \in P'$ performs a write operation $w_p$ in $C$.  
Moreover, assume that another process $q \neq p$ with $q \not \in P'$ is covering the same register that $w_p$ writes to.
Then if $\gamma \in \lit{Res}(C', P')$, then $w_p \gamma$ is in $\lit{Res}(C, P)$ where $P = P' \cup \{p\} \cup \{q\}$. 
\end{claim}
\subsubsection{New definition of valency}
We say that a configuration $C$ is $0$-valent$_U$ with respect to the set of processes $U$,
  if there exists a subset of at least $m+1$ processes $P \subseteq U$ and a reserving execution in $\lit{Res}(C, P)$ 
  that finishes when some process in $P$ returns $0$.
We call $C$ $0$-valent$_U^{m+1}$ w.r.t. $U$, 
  if there exists a subset of \emph{exactly} $m+1$ processes $P \subseteq U$ $(|P| = m+1)$,
  and a reserving execution interval in $\lit{Res}(C, P)$ returning $0$. 
We define $1$-valent$_U$ and $1$-valent$_U^{m+1}$ analogously.
If $U$ contains at least $m+1$ processes that have not returned,~\claimref{clm:reserved} implies that 
  every configuration is $0$-valent$_U^{m+1}$ or $1$-valent$_U^{m+1}$ (and thus $0$-valent$_U$ or $1$-valent$_U$).


As in our earlier definition in~\sectionref{sec:sqrt}, but unlike the standard definition,
  a configuration that is $0$-valent$_U^{m+1}$ can still also be $1$-valent$_U^{m+1}$
  in which case we call it bivalent$_U^{m+1}$.
Basically, a configuration is bivalent$_U^{m+1}$ if it is both 
  $0$-valent$_U^{m+1}$ due to some $P \subseteq U$ and 
  $1$-valent$_U^{m+1}$ due to some $Q \subseteq U$.
A configuration that is not bivalent$_U^{m+1}$ is called univalent$_U^{m+1}$.
Finally, similar to our earlier convention,
  we define a configuration to be $0$-univalent$_U^{m+1}$ if it is $0$-valent$_U^{m+1}$ but not $1$-valent$_U^{m+1}$.
On the other hand, a configuration that is $1$-valent$_U^{m+1}$ but not $0$-valent$_U^{m+1}$
  is called $1$-univalent$_U^{m+1}$.
Terms bivalent$_U$, univalent$_U$, $0$-univalent$_U$ and $1$-univalent$_U$ are defined analogously.

Next we prove a claim that lets us find reserving executions consisting of disjoint processes. 
\begin{claim}
\label{clm:disjoint}
Consider a configuration $C$ which is bivalent$_U$ w.r.t. $U$.
Assume that there are (possibly intersecting) sets of at least $m+1$ processes each $P \subseteq U$ and $Q \subseteq U$ 
  such that $|U| \geq |P| + |Q| + m$,
  and some reserving execution in $\lit{Res}(C, P)$ ends when $p \in P$ returns $0$,
  while some reserving execution in $\lit{Res}(C, Q)$ ends when $q \in Q$ returns $1$.
Then there are also disjoint sets of processes $P' \subseteq U$ and $Q' \subseteq U$ ($P' \cap Q' = \emptyset$), 
  such that an execution in $\lit{Res}(C, P')$ returns $0$ and an execution in $\lit{Res}(C, Q')$ returns $1$. 
Moreover, $m+1 \leq \min(|P'|, |Q'|) \leq \min(|P|, |Q|)$ and $\max(|P'|, |Q'|) \leq \max(|P|, |Q|)$. 
\end{claim}
\begin{proof}
None of the processes in $U$ may have already returned in configuration $C$,
  as that would contradict the existence of a reserving execution returning the other output. 
If $P$ and $Q$ do not intersect then we set $P' = P$ and $Q' = Q$.
Otherwise, we can find a set $H \subseteq U \setminus P \setminus Q$ of $m+1$ processes.
By~\claimref{clm:reserved}, $\lit{Res}(C, H)$ is non-empty, 
  and without loss of generality, some execution in $\lit{Res}(C, H)$ returns $0$.
Then, we set $P' = H$ and $Q' = Q$ (if all executions in $\lit{Res}(C, H)$ return $1$, we would set $P' = P$ and $Q' = H$).
\end{proof}
\subsubsection{The process-clone pairs and the proof}
As mentioned earlier, it is obviously not sufficient to simply cover registers with existing processes
  without any knowledge of what they are about to write.
In the proof of~\lemmaref{lem:sqrt} we used new clones that covered registers to block-overwrite these registers 
  back to the contents whose valency we knew.
In order to do something similar with existing processes, we associate a dedicated clone to each process.
The process and its clone remain in the same states and perform the same operations during the whole execution.

Usually, when we schedule a process to perform an operation,
  its clone performs the same operation immediately after the process.  
  Thus the pair of the process and the clone remain in the same state.
Under these circumstances, we can treat the pair of the process and its clone as a single process,
  because no process can distinguish the execution from when the clone would not take steps.
However, sometimes we will \emph{split} the pair by having only the process perform a write operation 
  and let the clone cover the register.
We will explicitly say when this is the case.
After we split the pair of process and clone in such a way, 
  we will not schedule the process to take any more steps and thus the clone will remain poised to write to the covered register.
After some delay, we will schedule the clone of the process to write, 
  effectively resetting the register to the value it had when the process wrote. 
Moreover, because meanwhile the process did not take any steps,
  after the write the clone will again be in the same state as its associated process.
Hence the pair of the process and clone will no longer be split, 
  and will continue taking steps in sync like a single process.

This is different from the way clones were used in the proof of~\lemmaref{lem:sqrt},
  because after the pair of the process and its clone is united, it can be split again.
Therefore, the same clone can reset the contents of registers written by its associated process multiple times, 
instead of requiring a new clone every time.

We call a split pair of a process and a clone \emph{fresh} as long as the register that 
  the process wrote to, and its clone is covering, has not been overwritten.
After the register is overwritten, we call the split pair \emph{stale}.

In addition, we also use cloning in a way similar to the proof of~\lemmaref{lem:sqrt}, 
  except that we do this at most constantly many times, as opposed to $r$ times, to reach the next configuration $C_{r+1}$.
Moreover, each time when we do this, we create duplicates of both the process and its corresponding clone.
This new process-clone pair is in the same state as the original pair, 
  and from there on behaves like a single new process similar to all other pairs.
We will always consider valency with respect to sets of processes whose pairs are not split. 
Therefore, the definition of valency does not need to change 
  when the clones keep taking steps immediately after their processes.

Sometimes, when considering process-clone pairs, none of which are split, 
  we may refer to them as processes, i.e. we may talk about a process taking steps or returning a value. 
As mentioned earlier, it is assumed that as long as the pair is not split, 
  the clone always follows and takes the same steps right after the process. 
Hence, in this context, a process taking a step means a pair taking a step.   

Now we are ready to prove the main result.
\begin{theorem}
\label{thm:linclone}
In the system of anonymous processes, 
  consider any correct consensus algorithm satisfying nondeterministic solo termination,
  with the property that every execution uses at most $m$ registers.
For each $r$ with $0 \leq r \leq m$, there exists a set $U$ containing $5m+6+2r$ process-clone pairs  
  such that a configuration $C_r$ is reachable through an execution $E_r$ 
  by processes and clones in $U$ with the following properties:
\begin{itemize}[noitemsep, nolistsep]
\item[1.] There exists a set $R$ of $r$ registers, that can be partitioned in two disjoint subsets $R = R_s \cup R_c$, where:
\begin{itemize}[noitemsep, nolistsep]
\item $R_s$ consists of all registers in the system that each have one fresh split pair on them, 
  last written by some process whose clone has not yet performed the write and is covering the register.
\item $R_c = R \setminus R_s$.
Each register in $R_c$ is covered by an unique pair of both a process and its clone.
\end{itemize}
Thus, each fresh pair is split on a different register in $R_s$,
  and an additional $|R_c|$ pairs are covering the registers in $R_c$.
Let $V$ be the set of these $|R_s| + |R_c| = r$ pairs.
\item[2.] There are at most $r$ stale split pairs in the system,
  that are all split on pairwise different registers from $R$.
Let $L$ be the set of these at most $r$ stale split pairs.
\item[3.] There exist disjoint sets of process-clone pairs that are not split 
  $P, Q \subseteq U \setminus V \setminus L$ with $|P| + |Q| \leq 2m+4$,
  such that an execution in $\lit{Res}(C_r, P)$ returns $0$ and an execution in $\lit{Res}(C_r, Q)$ returns $1$.\footnote{The pairs of processes in $P$ and $Q$ are not split, because all split pairs belong to $V \cup L$ (fresh to $V$ and stale to $L$). Also, the third condition implies that the configuration $C_r$ is bivalent$_{U \setminus V \setminus L}$.} 
\end{itemize}
\end{theorem}
\begin{proof}
The proof is by induction on $r$, with the base case $r = 0$.
Out of the $5m + 6$ processes-clone pairs, 
  half of them start with an input $0$ and half start with an input $1$.
We let $C_0$ be the initial state,
  $P$ be a set of some $m+1$ pairs with input $0$, and
  $Q$ be a set of some $m+1$ pairs with input $1$.
The first two properties are trivially satisfied; also $P \cap Q = \emptyset$ and $|P| + |Q| = 2m+2$.
By~\claimref{clm:reserved} and correctness of consensus,
  there is a reserving execution in $\lit{Res}(C_0, P)$ that decides $0$,
  and a reserving execution in $\lit{Res}(C_0, Q)$ that decides $1$ ($C_0$ is bivalent$_U$).
Observe that the pairs are not split and for the purposes of valency we can just consider the steps of processes.

Now, let us assume induction hypothesis for some $r$, i.e. the existence of $E_r$ and $C_r$ with the required three properties,
  and prove the step for $r+1$ by extending $E_r$ to $E_{r+1}$, resulting in the configuration $C_{r+1}$.
Let $U$, $P$, $Q$, $V$, $L$ and $R = R_s \cup R_c$ all be defined as in the theorem statement for $r$. 
  Our goal is to construct sets $U'$, $P'$, $Q'$, $V'$, $L'$ and $R' = R_s' \cup R_c'$ for $r+1$. 
In $U' \setminus U$ we have two more process-clone pairs available that have not taken steps and 
  can be used to clone an existing process-clone pair.
Let $T$ denote $U \setminus V \setminus L \setminus P \setminus Q$.
Since $|V| = r$, $L \leq r$ and $|P| + |Q| \leq 2m+4$, we have $|T| \geq 3m+2$.

For all but $|R_s| + |L|$ split pairs both processes and clones are in the same states, 
  about to perform the same operations.
By definition, each stale pair in $L$ is split on a different register from $R$.
In the following argument, we extend the execution from $E_r$ to $E_{r+1}$ by steps of processes and clones not in $L$.
This can introduce new stale split pairs and the resulting configuration $C_{r+1}$ 
  may not immediately satisfy the second property.
We will then show how to modify the extension and unite some stale split pairs, 
  such that the resulting configuration satisfies all properties, including the second property with the new $L'$.
	
Let $\alpha \in \lit{Res}(C_r, P)$ be the reserving execution interval that returns $0$,
  and let $\beta \in \lit{Res}(C_r, Q)$ be the reserving execution interval that returns $1$.
Notice that each time a process in $P$ or $Q$ takes a step in $\alpha$ or $\beta$, 
  its clone performs an identical step immediately after.
The execution $E_r \alpha$ ends with a process-clone pair $p \in P$ returning $0$ and 
  the execution $E_r \beta$ ends with a process-clone pair $q \in Q$ returning $1$. 

Each register in $R_c$ was covered by some pair of both a process and its clone in $V$.
Let $\gamma_c$ be a block write to all registers in $R_c$ by only the processes but not the clones 
  of these respective covering pairs: i.e. after each write we get a new fresh split pair.
Consider a configuration $D$ reached from $C_r$ by executing this block write, i.e. a configuration reached after $E_r \gamma_c$.
Assume that $D$ is $1$-valent$_T^{m+1}$, without loss of generality, because it has a valency.
For any execution interval $e$, let us denote by $W(e)$ the set of registers written to during $e$.
Hence, $R_s \cap W(e)$ is the set of registers in $R_s$ that are written-to during $e$.
Each register in $R_s$ is covered by a clone of a split pair whose process has already performed the write and is stopped.
Define $\gamma_s(e)$ as a block write to all registers in $R_s \cap W(e)$ by these trailing clones of the split pairs in $V$: 
  i.e. after each write another clone catches up with its process and a previously split pair is united.
Basically, if we run an execution interval $e$ from $C_r$ that changes contents of some registers in $R_s$,
  we can then clean these changes up by executing $\gamma_s(e)$, 
  which leads to all registers in $R_s$ having the same contents as in $C_r$.

Using a crude covering argument we can first show that
\begin{repclaim}{clm:alphaout}
The execution interval $\alpha$ must contain a write operation outside $R$.
\end{repclaim}
\noindent Based on this we can write $\alpha = \alpha' w_p \alpha''$, 
  where $w_p$ is the write operation to a register $\lit{reg} \not \in R$, performed by some process-clone pair $p \in P$.
    
Looking ahead, when we reach $C_{r+1}$, the new set of registers $R'$ will be $R \cup \{\lit{reg}\}$.
Next, we prove the following claim using an FLP-like case analysis:
\begin{repclaim}{clm:casean}
We can extend execution $E_r$ (i.e. from $C_r$) with an execution interval $e$ and reach a configuration 
  satisfying the first and the third inductive requirements to be $C_{r+1}$ 
  with a properly defined $U'$, $P'$, $Q'$, $V'$ and $R' = R_s' \cup R_c'$,
  and with all process-clone pairs that are not split being in sync.
But the second property is not immediately satisfied.
All stale split pairs from $L$ remain stale and split, but some pairs that were fresh and split on registers 
  in $R_s \cap W(e)$ may have become stale in $C_{r+1}$ 
  (because neither the process nor the clone in the split pair has taken steps while the register was overwritten in $e$). 
However, these are the only possible new stale split pairs in $C_{r+1}$, and they do not belong to the new sets $V' \cup P' \cup Q'$.
\end{repclaim}
The proofs of these claims are provided later. 
In order to finish the proof of the theorem, we need to show how to construct $L'$.
According to the above~\claimref{clm:casean} we can extend the execution to reach the next configuration $C_{r+1}$ 
  satisfying first and third but not the second property about the stale split pairs $L'$.
In $C_r$ we had at most $r$ stale pairs in the system, each split on a different register, and $L$ was the set of these pairs.
But on the way to reaching $C_{r+1}$, we may have introduced new stale pairs in the system.
According to~\claimref{clm:casean} these must be the pairs that were fresh and split on registers in $R_s \cap W(e)$ in $C_r$,
  and whose associated register in $R_s$ has been overwritten during $e$, making them stale in $C_{r+1}$.

The set of all stale pairs in $C_{r+1}$ may not satisfy the requirements imposed for $L'$,
  since there could already have been a stale pair split on a register in $R_s \cap W(e)$ in $L$ (in $C_r$).
Then two stale pairs would be split on this register in $C_{r+1}$, violating the second property.
However, for each such register in $R_s \cap W(e)$, 
  we know a stale pair $\rho \in L$ was split on it in $C_r$,
  and that this register was written-to during extension $W(e)$.
We now modify the extension $e$; we add a single write by the clone of the stale split pair $\rho$
  immediately before a write operation to the same register that was already in $e$.
This way, no pair other than the clone of $\rho$ observes a difference between the two executions, 
  and we will use the configuration reached by the modified execution as $C_{r+1}$.
Because of this indistinguishability, the new $C_{r+1}$ still satisfies other required properties.
Moreover, the pair $\rho$ is not split anymore; it is united since the clone has caught up with its process.

We can do the above modification to the execution for each register in $R_s \cap W(e)$ 
  that previously ended up with two stale split processes in $C_{r+1}$.
Let the modified execution extension be $e'$.
In $e'$, some stale split pairs from $L$ are united, indistinguishably to all other processes and clones, 
  leading to a configuration $C_{r+1}$, that still satisfies the first and third properties,
  and has at most one stale pair split on any register.
We take $L'$ to be the set of stale split pairs.
By construction, all stale pairs are split on registers in $R'$ and no two on the same register, 
  so we do have $|L'| \leq r+1$ as desired.
Hence, we have reached configuration $C_{r+1}$ satisfying all properties and completing the proof.
\end{proof}
\begin{corollary}
In a system of $n$ anonymous processes, 
  any consensus algorithm satisfying non-deterministic solo termination must use $\Omega(n)$ registers.
\end{corollary}
\begin{proof}
\theoremref{thm:linclone} directly implies the $\Omega(n)$ lower bound on the number of registers used in some execution.
If $n$ is the number of anonymous processes and no execution uses more than $m = n/20$ registers,
  by~\theoremref{thm:linclone} we can reach $C_{m}$ for large enough $n$, and we have enough processes $n \geq 10m + 12 + 4m$.
In $C_{m}$ there are $m$ registers in $R$, each of which has either already been written-to ($R_s$)
  or are covered by unique processes ($R_c$).
We could perform a block write to $R_c$ by covering processes from $V$ in $C_{m}$, 
  after which in the resulting execution $m = n/20 = \Omega(n)$ different registers would have been written to.
\end{proof}
We now provide the delayed proofs for the claims.
\begin{claim}
\label{clm:alphaout}
The execution interval $\alpha$ must contain a write operation outside $R$.		 
\end{claim}
\begin{proof}
Assume the contrary.
We know that the execution $E_r \alpha$ decides $0$.
No process or clone that takes a step in $\gamma_c$ or $\gamma_s(\alpha)$ appears in $\alpha$
  (they belong to $V$, disjoint from $P$ and $Q$),  
  and by definition, no process or clone from $T$ takes a step in $\alpha$, $\gamma_c$ or $\gamma_s(\alpha)$. 
Thus, to all processes (and clones) in $T$, 
  the configurations after $E_r \alpha \gamma_s(\alpha) \gamma_c$ and after $E_r \gamma_c$,
  which is configuration $D$, are indistinguishable.
This is because no process (or clone) in $T$ has taken steps, the registers in $R$ contain the same values, 
  and other registers were not touched during $\alpha$, $\gamma_s(\alpha)$ or $\gamma_c$.
Configuration $D$ is $1$-valent$_T^{m+1}$, so some extension from $E_r \alpha \gamma_s(\alpha) \gamma_c$ by 
  an execution interval from $\lit{Res}(D, T)$ decides $1$.
This contradicts the correctness of the algorithm.
\end{proof}
\input{flpcases.tex}

%% file: flpcases.tex
\begin{claim}
\label{clm:casean}
We can extend execution $E_r$ (i.e. from $C_r$) with execution interval $e$ and reach a configuration 
  satisfying the first and the third inductive requirements to be $C_{r+1}$ 
  with properly defined $U'$, $P'$, $Q'$, $V'$ and $R' = R_s' \cup R_c'$,
  and with all process-clone pairs that are not split being in sync.
But the second property is not immediately satisfied.
All stale split pairs from $L$ remain stale and split, but some pairs that were fresh and split on registers 
  in $R_s \cap W(e)$ may have become stale in $C_{r+1}$ 
  (because neither the process nor the clone in the split pair has taken steps while the register has been overwritten in $e$). 
However, these can be the only new stale split pairs in $C_{r+1}$ and they do not belong to the new sets $V' \cup P' \cup Q'$.
\end{claim}
\begin{proof}
The proof works by case analysis. We use the notation from~\theoremref{thm:linclone}.
$T$ does not contain any split pairs, so we can consider valency with respect to processes in $T$.  

\paragraph{Case 1: the configuration reached by the execution $E_r \alpha'$ is $1$-valent$_T^{m+1}$:}
Let $\ell$ be the length of $w_p \alpha''$ and $A_j$ be a prefix of $w_p \alpha''$ of size $j$ for $0 \leq j \leq \ell$.
Here we consider steps of process-clone pairs from $P$, 
  i.e. the difference between $A_j$ and $A_{j+1}$ is the same operation performed twice by a process and its clone,
  and $l$ counts these couples of identical operations,
  as illustrated in~\figureref{fig:case1}.
Pairs in $P$ are not split, as by the inductive hypothesis only pairs in $V \cup L$ are split, 
  and $(V \cup L) \cap P = \emptyset$.

\begin{figure}
\label{fig:case1}
\caption{Case 1}
\resizebox{\textwidth}{!}{\includegraphics{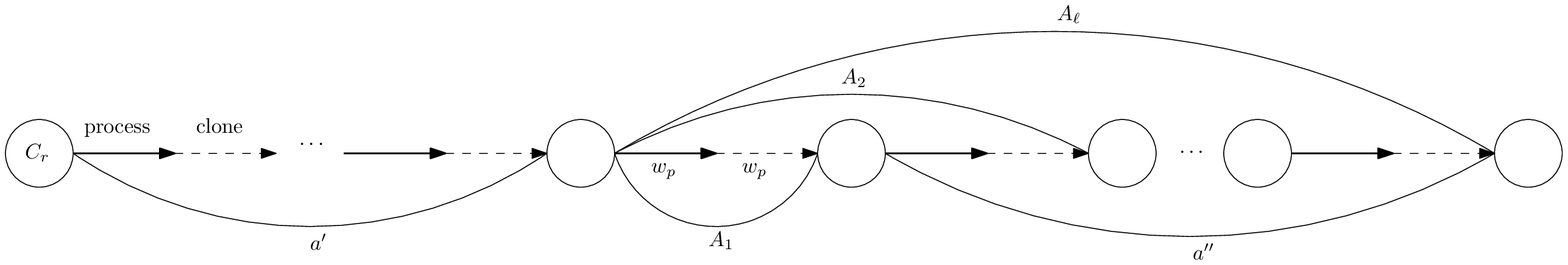}}
\end{figure}

We consider two further subcases.

\paragraph{Case 1.1: for some $0 \leq j \leq \ell$, 
  the configuration reached by the execution $E_r \alpha' A_j$ is bivalent$_T^{m+1}$:}
In this case, we let $C_{r+1}$ be precisely the configuration reached by $E_{r+1} = E_r \alpha' A_j$, 
  and let $R' = R \cup \{\lit{reg}\}$.
Clearly, $|R'| = r+1$.
Moreover, since $C_{r+1}$ is bivalent$_T^{m+1}$ there are two subsets of $m+1$ process-clone pairs in $T$ 
  and respective reserving execution intervals from $C_{r+1}$ that return $0$ and $1$.
Since $|T| \geq 3m+2$, by~\claimref{clm:disjoint}, we can actually find $P' \subseteq T$ and $Q' \subseteq T$, 
  with $P' \cap Q' = \emptyset$ and $|P'| = |Q'| = m+1$, such that
  an execution in $\lit{Res}(C_{r+1}, P')$ returns $0$ and an execution in $\lit{Res}(C_{r+1}, Q')$ returns $1$.
These sets $P'$ and $Q'$ are the new sets of pairs for $C_{r+1}$, and as required $|P'| + |Q'| = 2m+2 \leq 2m+4$.

The new set $R_s'$ will be $R_s \setminus (R_s \cap W(\alpha' A_j))$, i.e. the registers from $R_s$ 
  that have not been touched during our execution extension $\alpha' A_j$ from $C_r$ to $C_{r+1}$.
For each of these registers we still have the same one pair from $V$ split on it, 
  and since this pair was fresh in $C_r$ and the register has not been written to during the extension,
  it is still fresh in $C_{r+1}$ as required.  
There are no other fresh split pairs in $C_{r+1}$:
  no new split pairs were introduced during the extension,
  and the rest of fresh pairs in $C_r$ were split on $R_s \cap \alpha' A_j$.
These pairs are not fresh in $C_{r+1}$, as their register was written to during the extension.

The set $R_c'$ is simply $(R \cup \{\lit{reg}\}) \setminus R_s' = R_c \cup (R_s \cap W(\alpha' A_j)) \cup \{\lit{reg}\}$. 	
We must show that there is a unique process-clone pair covering each of these registers.
For each register in $R_c$, we take the same pair from $V$ that was associated and covering it in $C_r$.
For each register in $(R_s \cap W(\alpha' A_j)) \cup \{\lit{reg}\}$, 
  we find a unique pair from $P$ covering it in $C_{r+1}$.
Since $\alpha$ is a reserving execution interval from $C_r$, 
  all its prefixes including $\alpha' A_j$ are also reserving.
Thus, in $C_{r+1}$ for each register that has been ever written to during $\alpha' A_j$, 
  in particular for registers in $(R_s \cap W(\alpha' A_j)) \cup \{\lit{reg}\}$, 
  we find and associate a unique covering pair in $P$.
Technically, if $j=0$, register $\lit{reg}$ is not yet written, 
  but the next operation in $\alpha$ is $w_p$ by a pair covering $\lit{reg}$.

The set $V'$ contains all $r+1$ pairs that we associated with registers in $R'$, so $V' \subseteq V \cup P$.
The number of stale split pairs may have increased, however.
We still have pairs in $L$ plus $|R_s \cap W(\alpha' A_j)|$ of the previously fresh pairs that are now stale.
We deal with this in~\theoremref{thm:linclone} by modifying the execution extension to unite some stale pairs from $L$, 
  leaving us with a desired subset $L'$ of at most $r+1$ stale pairs.

Finally, as $P', Q' \subseteq T$, and $T$ was disjoint from $V$, $P$, $Q$ and $L$,
  $(P' \cup Q') \cap (V' \cup L') = \emptyset$ as required.

\paragraph{Case 1.2: for every $0 \leq j \leq \ell$, 
  the configuration reached by the execution $E_r \alpha' A_j$ is univalent$_T^{m+1}$:}
By \emph{\textbf{Case 1}} assumption $E_r \alpha' A_0$ is $1$-valent$_T^{m+1}$, 
  so it must be $1$-univalent$_T^{m+1}$.
On the other hand, $\alpha$ ends with a process (and its clone) returning $0$, 
  so the configuration reached by $E_r \alpha' A_{\ell}$ must be $0$-univalent$_T^{m+1}$.
No intermediate configuration is bivalent, so we can find a $0 \leq j < \ell$ such that 
  $E_r \alpha' A_j$ leads to a $1$-univalent$_T^{m+1}$ configuration and 
  $E_r \alpha' A_{j+1}$ leads to a $0$-univalent$_T^{m+1}$ configuration. 
We take $C_{r+1}$ to be the configuration reached after $E_{r+1} = E_r \alpha' A_j$,
  and define sets $R_c'$, $R_s'$, $V'$ and $L'$ the same way as if $C_{r+1}$ was bivalent in \emph{\textbf{Case 1.1}}.
This still works, but we need a new way to find $P'$ and $Q'$ with desired properties.

Let $o$ be the operation by a process-clone pair in $P$ separating $A_j$ and $A_{j+1}$.
$o$ may not be a read, as no process (or clone) in $T$ can distinguish between configurations 
  after $E_r \alpha' A_j$ or $E_r \alpha' A_j o$,
  making it impossible for these configurations to have different univalencies w.r.t. to $T$.
Let $Q'$ be a set of any $m+1$ pairs in $T$. 
By~\claimref{clm:reserved}, $\lit{Res}(C_{r+1}, Q')$ is non-empty and since $C_{r+1}$ is $1$-univalent$_T^{m+1}$,
  all executions in $\lit{Res}(C_{r+1}, Q')$ return $1$.
Recall that $U'$ is $U$ plus two process-clone pairs that have not taken any steps,
  so that $|U'| = |U| + 2 \leq 5m + 6 + 2(r+1)$. 
Let us use these process-clone pairs to create two duplicates of the process-clone pair performing the write operation $o$.
Both of these new pairs will be in the same state as the original pair performing $o$.
These duplicate processes (and their clones) are thus poised on the same register about to perform 
  write operations $o'$ and $o''$ identical to the operation $o$ at configuration $C_{r+1}$. 

Recall that $|T| \geq 3m+2$ and let $F$ be a set of $m+1$ pairs from $T \setminus Q'$.
Let $P'$ be the union of $F$ and the two new duplicate pairs, $|P'| = m+3$ in total.
Let $O$ be the $0$-univalent$_T^{m+1}$ configuration reached after $E_r \alpha' A_{j+1} = E_r \alpha' A_j o = E_{r+1} o$.
Due to $0$-univalency, there is a reserving execution $\xi \in \lit{Res}(O, F)$ that returns $0$.
Having one duplicate pair perform $o'$ from $C_{r+1}$ while another covers the same register with the same operation $o''$, 
  we reach the state indistinguishable from $O$ for all $m+1$ pairs in $F$.
Thus, execution $o'\xi$ from $C_{r+1}$ returns $0$, and by~\claimref{clm:prefix}, $o'\xi \in \lit{Res}(C_{r+1}, P')$.
By construction $|P'| + |Q'| = 2m+4$ and $P' \cap Q' = \emptyset$, as desired and 
  the intersection of $P'$ and $Q'$ with $V'$ or $L'$ is empty like in \emph{\textbf{Case 1.1}}.

The remaining case is when the configuration reached by the execution $E_r \alpha'$ is $0$-univalent$_T^{m+1}$.
It is a bit more involved but utilizes the same general ideas and techniques.
One difference is that we also split pairs.
The construction is given in~\appendixref{app:case2}.
\end{proof}

%% file: adaptive.tex
\section{Extensions}
\paragraph{Adaptive Lower Bound:}
Let us sketch a proof for an adaptive linear lower bound on the space complexity of consensus for non-anonymous processes
  but under extra restrictions on register size and solo termination.
In this setting, processes are no longer anonymous, but we assume they come from a very large namespace.
Each of these huge number of processes executes its own code,
  however, we get to choose which subset of processes participates in the execution.
We show that there is a linear space lower bound that depends on the number of participating processes, 
  that is, for large enough namespace, we can find an execution of $n$ processes (out of all processes) 
  where $\Omega(n)$ registers get written.
	
The restrictions are that the registers have a bounded size and 
  that the consensus algorithm satisfies bounded nondeterministic solo termination property,
  meaning that there always is a terminating solo execution of a process consisting of less than certain number of steps.
If we had bounded nondeterministic solo termination, the lower bound execution for anonymous processes 
  constructed in~\theoremref{thm:linclone} would always contain less than $B$ steps,
  where $B$ is a finite bound that only depends on $n$ and the solo termination bound.
As registers have a bounded size, for both input values, 
  a process can exhibit only finitely many different behaviors during its first $B$ steps, 
  because in each step it can either read or write a fixed number of different values.
For a sufficiently large namespace (depending on $B$, $n$ and register size), 
  by pigeon-hole principle, we can find $n$ processes such that half of them start with input $1$, half start with $0$
  and all processes with the same input behave as anonymous for the first $B$ steps of an execution.
Hence, we can use~\theoremref{thm:linclone} and get an execution where $n/20$ registers are written to,
  as described at the end of~\sectionref{sec:linbound}.
	
\paragraph{Future Work:}
We believe that is should be possible to derive the above adaptive lower bound without the bounded solo termination assumption,
  and to get good estimate on the required size of the namespace.
However, the major open problem is still to resolve the general, non-anonymous and non-adaptive case,
  i.e. to get tight bounds on the space required to solve consensus with exactly $n$ asymmetric processes.

%% file: appendix.tex
\newpage
\section{Last case of the proof of~\claimref{clm:casean}}
\label{app:case2}
The notation here is the same as in~\theoremref{thm:linclone}.
In particular, we use the definitions of configuration $D$ and execution intervals $\gamma_c$ and $\gamma_s(\alpha')$.

\begin{figure}
\label{fig:case2}
\caption{Case 2}
\resizebox{\textwidth}{!}{\includegraphics{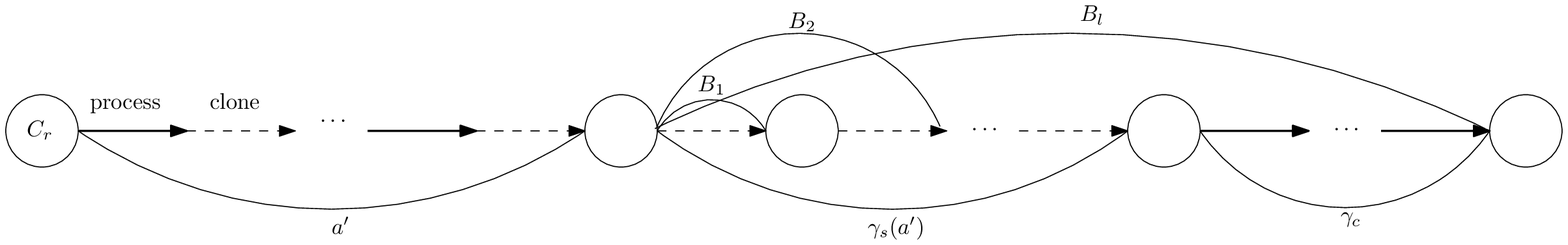}}
\end{figure}

\paragraph{Case 2: the configuration reached by the execution $E_r \alpha'$ is $0$-univalent$_T^{m+1}$:}
Let $D'$ be the configuration reached after executing $E_r \alpha' \gamma_s(\alpha') \gamma_c$.
Recall that in $\gamma_s(\alpha')$, the clones of split pairs overwrite 
  all registers in $R_s \cap W(\alpha')$ with the values these registers had in $C_r$.
Configuration $D'$ is indistinguishable from the $1$-valent$_T^{m+1}$ configuration 
  $D$ reached by $E_r \gamma_c$ for all processes and clones in $T$.
This is because the processes (and clones) in $T$ have not taken steps and 
  the contents of all registers are the same in $D$ and $D'$ since $\alpha'$ contains writes only to registers in $R$.

Let us denote by $l$ the length of execution interval $\gamma_s(\alpha') \gamma_c$ and 
  let $B_j$ be the prefix of this execution interval of size $0 \leq j \leq l$.
Unlike the definition of $A$, when the difference between $A_j$ and $A_{j+1}$ was a ``pair step,'' 
  i.e. two identical operations of the process and its respective clone,
    the difference between $B_j$ and $B_{j+1}$ is actually a single step, i.e. exactly one operation
    by either a process or a clone of some process.
This is because by definition each of the steps in $\gamma_s(\alpha')$ is performed only by a clone 
  (uniting a previously split pair after each operation), 
  while each step in $\gamma_c$ is performed only by a processes 
  (creating a new split pair after each operation).
This is illustrated in~\figureref{fig:case2}.

\paragraph{Case 2.1: for some $0 \leq j \leq l$, 
  the configuration reached by the execution $E_r \alpha' B_j$ is bivalent$_T^{m+1}$:}
We let $C_{r+1}$ be the configuration reached by $E_{r+1} = E_r \alpha' B_j$, 
  and $R' = R \cup \{\lit{reg}\}$, with $|R'| = r+1$.
Since $C_{r+1}$ is bivalent$_T^{m+1}$ there are two subsets of process-clone pairs in $T$ 
  and respective reserving execution intervals from $C_{r+1}$ that decide different outputs.
By~\claimref{clm:disjoint}, we can find $P' \subseteq T$ and $Q' \subseteq T$, 
  with $P' \cap Q' = \emptyset$ and $|P'| = |Q'| = m+1$, such that
  an execution in $\lit{Res}(C_{r+1}, P')$ returns $0$ and an execution in $\lit{Res}(C_{r+1}, Q')$ returns $1$.
As in \emph{\textbf{Case 1.1}} these sets $P'$ and $Q'$ are the new sets of pairs for $C_{r+1}$.

The new set $R_s'$ is $(R_s \setminus (R_s \cap W(\alpha'))) \cup (R_c \cap W(B_j))$, consisting of 
  registers from $R_s$ not touched during $\alpha'$ and 
  registers from $R_c$ written to during the prefix of block write $\gamma_c$ that was executed in $B_j$.
For each register in $R_s \setminus (R_s \cap W(\alpha'))$ 
  we still have the same one pair from $V$ split on it as in $C_r$.
This split pair was fresh in $C_r$ and since its register in $R_s$ has not been written during the extension $\alpha' B_j$
  (not written in $\alpha'$, and hence also not in $\gamma_s(\alpha')$ or $\gamma_c$),
  it is still fresh as required in $C_{r+1}$.
New fresh split pairs are created during the execution of the prefix of $\gamma_c$ in $B_j$,
  as only processes but not their clones take steps and after each write in $\gamma_c$ we get a new fresh split pair.
These pairs are split on registers in $R_c \cap W(B_j)$,
  and we associate exactly one new fresh split pair to each of these registers.  
No other split pairs are fresh, since fresh pairs that were split on $R_s \cap \alpha'$ cannot be fresh in $C_{r+1}$, 
  as their registers were written to during $\alpha'$.

The set $R_c'$ is 
  $(R \cup \{\lit{reg}\}) \setminus R_s' = (R_c \setminus (R_c \cap W(B_j))) \cup (R_s \cap W(\alpha')) \cup \{\lit{reg}\}$.
As in \emph{\textbf{Case 1.1}}, $\alpha'$ is a prefix of a reserving execution interval $\alpha \in \lit{Res}(C_r, P)$,
  so for each register in $R_s \cap W(\alpha')$ we can find a unique covering process (and thus respective pair) 
  from $P$ in $C_{r+1}$.
The register $\lit{reg}$ is covered by the process-clone pair in $P$ with a pending write $w_p$.
For each register in $R_c \setminus (R_c \cap W(B_j))$, 
  we take the pair from $V$ that was associated and covering it in $C_r$.
Neither process nor clone in this pair have taken any steps in $\alpha' B_j$ 
  and are still covering the same register in $C_{r+1}$.

The set $V'$ contains all $r+1$ pairs associated with registers in $R'$, and $V' \subseteq V \cup P$.
The number of split stale pairs has again increased from $L$ by at most $|R_s \cap W(\alpha')|$ 
  due to previously fresh pairs split on $R_s \cap W(\alpha')$ that may now be stale.
Also, as in \emph{\textbf{Case 1.1}}, $(P' \cup Q') \cap (V' \cup L') = \emptyset$ as required.
	
\paragraph{Case 2.2: for every $0 \leq j \leq l$, 
  the configuration reached by the execution $E_r \alpha' B_j$ is univalent$_T^{m+1}$:}
This case is similar to \emph{\textbf{Case 1.2}}.
The configuration reached after $E_r \alpha' B_0$ is $0$-univalent$_T^{m+1}$ by the \emph{\textbf{Case 2}} assumption and 
  configuration $D'$ reached after $E_r \alpha' B_l$ must be $1$-univalent$_T^{m+1}$
  as it is indistinguishable from $1$-valent$_T^{m+1}$ configuration $D$ for all processes (and clones) in $T$.
Therefore, we can find a $0 \leq j < l$ such that 
  $E_r \alpha' B_j$ leads to a $0$-univalent$_T^{m+1}$ configuration and 
  $E_r \alpha' B_{j+1}$ leads to a $1$-univalent$_T^{m+1}$ configuration. 
We let $C_{r+1}$ be the configuration reached after $E_{r+1} = E_r \alpha' B_j$,
  and define $R_c'$, $R_s'$, $V'$ and $L'$ as if $C_{r+1}$ was bivalent in \emph{\textbf{Case 2.1}}.
This again works, and so does the claim about stale split pairs in $C_{r+1}$,
  but we have to construct $P'$ and $Q'$ with the desired properties.

However, we can construct $P'$ and $Q'$ in a very similar way to \emph{\textbf{Case 1.2}}.
If $o$ is the operation separating $B_j$ and $B_{j+1}$, $o$ may not be a read as before,
  and we again create two new duplicate process-clone pairs, both about to perform identical write operations $o'$ and $o''$.
We let $P'$ be a set of any $m+1$ pairs and 
  $Q'$ be a set of $m+3$ pairs with $m+1$ pairs from $T \setminus P'$ and two new duplicate pairs.
Then, by the same argument of \emph{\textbf{Case 1.2}}, $P'$ and $Q'$ satisfy all required properties.